\documentclass[a4paper,UKenglish,cleveref, autoref, thm-restate]{lipics-v2021}

\pdfoutput=1 
\hideLIPIcs  

\usepackage{xspace}
\usepackage{listings}
\usepackage{mathtools}

\usepackage[]{todonotes} 

\usepackage{tikz-cd}
\usetikzlibrary{calc, arrows.meta, decorations.pathreplacing, calligraphy, backgrounds}

\usepackage[frozencache=true,cachedir=mintedcache]{minted}
\setminted{autogobble,bgcolor=lightgrey,mathescape=true,escapeinside=//,fontsize=\footnotesize}

\usepackage{csquotes}

\definecolor{egreen}{rgb}{0, 0.4, 0.267}
\definecolor{dkviolet}{rgb}{0.6,0,0.8}
\definecolor{dkgreen}{rgb}{0,0.4,0}
\definecolor{dkblue}{rgb}{0,0.1,0.5}
\definecolor{lightblue}{rgb}{0,0.5,0.5}
\definecolor{orange}{rgb}{0.9,0.39,0}
\definecolor{lightgrey}{RGB}{240,240,240}

\newcommand{\namefont}[1]{\textsf{#1}}
\newcommand{\Coq}{\namefont{Coq}\xspace}

\newcommand{\MC}{\namefont{Mathematical Components}\xspace}

\newcommand{\ie}{i.e.}
\newcommand{\eg}{e.g.}

\newcommand{\N}{\mathbb N}

\newcommand{\Cat}{\mathcal C}
\newcommand{\G}{\mathcal G}

\newcommand{\A}{\mathcal A}
\newcommand{\Q}{\mathcal Q}
\newcommand{\BP}{\mathcal{BP}}

\newcommand{\E}{\mathcal E}

\newcommand{\PQ}{\mathrm{PQ}}

\newcommand{\card}[1]{\mathrm{card}(#1)}

\newcommand{\restr}{\mathrm{restr}}
\newcommand{\com}{\mathrm{commute}}
\newcommand{\id}{\mathrm{id}}

\newcommand{\sot}{\textsf{st}}

\newcommand{\rel}[1]{\sim}
\newcommand{\eqd}{\approx}

\newcommand{\dual}{\dagger}

\newcommand{\Comp}{\mathrm{Comp}}

\newcommand{\cl}{c\ell} 
\newcommand{\concat}{\!\cdot\!} 
\newcommand{\acc}{\mathrm{acc}} 
\newcommand{\arcsfrom}[1]{\A_{#1}} 
\newcommand{\arcsto}[1]{\;\overline{\!\!\A}_{#1}} 

\newcommand{\Mono}{\mathrm{Mono}}

\newcommand{\rawvertex}[1]{\tikz{\fill[#1] (0,0) circle (2pt);}}

\NewDocumentCommand{\quiver}{ O{1} m O{0} m O{0} }{
  \tikz[baseline={([yshift=\ifblank{#1}{-.6ex}{0ex}]current bounding box.center)}, inner sep=\ifblank{#1}{1pt}{.5pt}, -{Latex[scale=\ifblank{#1}{1}{.5}]}, scale=\ifblank{#1}{1}{.5}]{
    \bfseries
    \foreach \p [count = \j, evaluate={\i=int(\j-1)}, evaluate={\k=int(100-(\i<#3)*70)}] in {#2} {
      \node[black!\k] (\i) at \p {.};
    }
    \foreach \i/\j [count = \l, evaluate={\k=int(100-(\l<=#5)*70)}] in {#4} {
      \draw[black!\k] (\i)--(\j);
    }
  }
}

\NewDocumentCommand{\quiverr}{ O{1} m O{0} m O{0} }{
  \tikz[baseline={([yshift=\ifblank{#1}{-.5ex}{0ex}]current bounding box.center)}, inner sep=\ifblank{#1}{1pt}{.5pt}, -{Latex[scale=\ifblank{#1}{1}{.5}]}, scale=\ifblank{#1}{1}{.5}]{
    \bfseries
    \foreach \p [count = \j, evaluate={\i=int(\j-1)}, evaluate={\k=int(100-(\i<#3)*70)}] in {#2} {
      \node[black!\k] (\i) at \p {.};
    }
    \foreach \i/\j/\s [count = \l, evaluate={\k=int(100-(\l<=#5)*70)}] in {#4} {
      \draw[black!\k] (\i)to[bend left=\s](\j);
    }
  }
}

\newcommand{\Qmap}[1][1]{{\quiver[#1]{(0,0),(.8,0)}{0/1}}}

\newcommand{\Qcobimap}[1][1]{{\quiverr[#1]{(0,0),(1,0)}{0/1/30,0/1/-30}}}
\newcommand{\Qcomp}[1][1]{{\quiver[#1]{(0,0),(.5,-.5),(1,0)}{0/1,0/2,1/2}}}
\newcommand{\Qmono}[1][1]{{\quiverr[#1]{(0,0),(.8,0),(1.6,0)}{0/1/20,0/1/-20,0/2/35,1/2/0}}}

\newcommand{\mapicomp}{{\quiver{(1,0),(.5,-.5),(0,0)}[1]{1/0,2/0,2/1}[2]}}
\newcommand{\mapiocomp}{{\quiver{(.5,-.5),(0,-0),(1,0)}[1]{1/0,0/2,1/2}[2]}}
\newcommand{\mapiicomp}{{\quiver{(0,0),(.5,-.5),(1,0)}[1]{0/1,0/2,1/2}[2]}}

\newcommand{\mapioomono}  {{\quiverr{(0,0),(.8,0),(1.6,0)}[1]{0/1/20,0/1/-20,0/2/35,1/2/0}[3]}}
\newcommand{\compimono}   {{\quiverr{(0,0),(.8,0),(1.6,0)}{0/1/-20,0/1/20,0/2/35,1/2/0}[1]}}
\newcommand{\compiomono}  {{\quiverr{(0,0),(.8,0),(1.6,0)}{0/1/20,0/1/-20,0/2/35,1/2/0}[1]}}
\newcommand{\cobimapimono}{{\quiverr{(1.6,0),(.8,0),(0,0)}[1]{2/0/35,1/0/0,2/1/20,2/1/-20}[2]}}

\newcommand{\coq}[1]{\mintinline{coq}`#1`}
\newcommand{\coqescape}[1]{\mintinline[mathescape=true,escapeinside=//]{coq}`#1`}

\title{Machine-Checked Categorical Diagrammatic Reasoning} 

\author{Benoît Guillemet}{École normale supérieure Paris-Saclay, France}{}{}{}

\author{Assia Mahboubi}{Nantes Université, École Centrale Nantes, CNRS, INRIA, LS2N, UMR 6004, France \and Vrije Universiteit Amsterdam, the Netherlands}{}{0000-0002-0312-5461}{} 

\author{Matthieu Piquerez}{Nantes Université, École Centrale Nantes, CNRS, INRIA, LS2N, UMR 6004, France}{}{0009-0002-1126-4725}{}

\authorrunning{B. Guillemet, A. Mahboubi, M. Piquerez} 

\Copyright{Jane Open Access and Joan R. Public} 



\category{} 

\relatedversion{} 

\nolinenumbers 

\EventEditors{John Q. Open and Joan R. Access}
\EventNoEds{2}
\EventLongTitle{42nd Conference on Very Important Topics (CVIT 2016)}
\EventShortTitle{CVIT 2016}
\EventAcronym{CVIT}
\EventYear{2016}
\EventDate{December 24--27, 2016}
\EventLocation{Little Whinging, United Kingdom}
\EventLogo{}
\SeriesVolume{42}
\ArticleNo{23}

\begin{document}

\maketitle

\begin{abstract}
This paper describes a formal proof library, developed using the Coq proof assistant, designed to assist users in writing correct diagrammatic proofs, for 1-categories.  This library proposes a deep-embedded, domain-specific formal language, which features dedicated proof commands to automate the synthesis, and the verification, of the technical parts often eluded in the literature.
\end{abstract}

\section{Introduction}\label{sec:intro}

\emph{Abstract nonsense}, a non derogatory expression attributed to Steenrod, usually refers to the incursion of categorical methods for a proof step deemed both technical and little informative, and therefore often succinctly described. \emph{Diagrams} are typically drawn in this case, so as to guide the intuition of the audience, and help visualize the existence of certain morphisms or objects, identities between composition of morphisms, etc.

Formally, a categorical diagram is a functor $F\colon J \rightarrow \mathcal{C}$, with $J$ a small category called the \emph{shape} of the diagram~\cite{zbMATH01324388}. Diagrams are depicted as directed multi-graphs, also called \emph{quivers}, whose vertices are decorated with the objects of $\mathcal{C}$ and whose arrows each represent a certain morphism, between the objects respectively decorating its source and its target. A directed path in the diagram is hence associated with a chain of composable arrows and a diagram \emph{commutes} when all directed paths with same source and target lead to equal compositions. Equalities between compositions of morphisms thus correspond to the commutativity of certain sub-diagrams of a larger diagram. Chasing commutative sub-diagrams in a larger diagram provides an elegant alternative to equational reasoning, when the latter becomes overly technical. Diagrams actually play a central role in category theory, for they provide such an efficient way of delivering convincing enough proofs. Some classical textbooks introduce diagrams as early as in their introduction chapter~\cite{zbMATH01216133}, while others devote an entire section to diagrammatic categorical reasoning~\cite[Section 1.6]{riehl}~\cite[Session 17]{zbMATH05593534}. The following diagrammatic proof of Lemma~\ref{lem:monomonom} provides a toy illustrative example of this technique.

\begin{lemma}\label{lem:monomonom}
  Let $\Cat$ be a category. For any morphism $f$ and $g$ such that $g \circ f \in Hom(\Cat)$, if $g \circ f$ is a monomorphism, then so is $f$.
\end{lemma}

\begin{proof}
  Consider $f : A \rightarrow B$, $g : B \rightarrow C$ and $g \circ f : A \rightarrowtail C$ morphisms in a category $\Cat$. Here is a very detailed diagrammatic proof, taking place in the diagram of Figure~\ref{fig:monomonom}.

  \begin{figure}[h!]
    \begin{center}
      \begin{tikzcd}
      & & B \arrow[rd, "g"] & \\
      D \arrow[r, "k"', bend right] \arrow[r, "h", bend left] \arrow[rru, "j", bend left] \arrow[rrr, "g \circ j"', bend right=49] & A \arrow[ru, "f"] \arrow[rr, "g \circ f", tail] &                   & C
      \end{tikzcd}
    \end{center}
    \caption{Diagrammatic proof that if $g \circ f$ is a monomorphism, then so is $f$}\label{fig:monomonom}
  \end{figure}
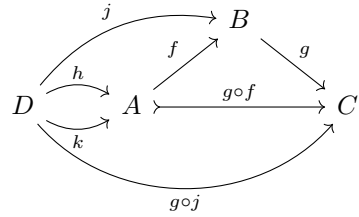

  We need to prove that for any two other morphisms $h,k : D \rightrightarrows A$ in $Hom(\Cat)$ such that $f \circ h = f \circ k$, we have $h = k$, i.e., that the two-arrows diagram $h, k$ commutes. By hypothesis on $h$ and $k$, there is an arrow $j$ such that the triangle diagrams respectively formed by arrows $h,f, j$ and $k, f, j$ both commute. By definition of composition, the triangle diagrams formed by arrows $f, g, g \circ f$ and $j, g, g \circ j$ also commute. As a consequence, both triangle diagrams formed by $h, g \circ f, g \circ j$ and $k, g \circ f, g \circ j$ commute. The conclusion follows because $g \circ f$ is a monomorphism.
\end{proof}

\emph{Diagram chasing} actually refers to a central technique to homological algebra~\cite{zbMATH00758278,zbMATH01324388}, operating on diagrams over abelian categories and used for proving the existence, injectivity, surjectivity of certain morphisms, the exactness of some sequences, etc. The \emph{five lemma} or the \emph{snake lemma} are typical examples of results proved by diagram chasing~\cite{zbMATH00758278}. However, complex diagram chases (see, e.g., \cite[p.338]{piquerez:tel-03499730}) only remain readable at the price of hiding non-trivial technical arguments and are, in practice, challenging to rigorously verify by hand. Typically, the reader of a diagrammatic proof is asked to solve instances of variable difficulty of a decision problem hereafter referred to as the \emph{commerge} problem: \emph{Given a collection of sub-diagrams of a larger diagram which commute, must the entire diagram commute?} In addition, proofs may resort to non-trivial duality arguments, in which the reader has to believe that a property about diagrams in any Abelian category remains true after reversing all the involved arrows, although the replay of a given proof \emph{mutatis mutandis} cannot be fulfilled in general.

The long-term objective of the present work is thus to build a computer-aided instrument for devising both fluent and reliable categorical diagrammatic reasoning, for 1-categories. The present article describes the implementation of the core of such a tool, as a library for the \Coq{} proof assistant~\cite{the_coq_development_team_2023_8161141}. The design of this library follows two main design principles. First, it aims at being independent from any specific library of formalized category theory or abstract algebra, but rather usable as a helper for any existing one. Second, it should feature enough automation tools for synthesizing the bureaucratic parts of proofs \enquote{by abstract nonsense}, and formal proofs thereof. The main contributions presented here are thus:
\begin{itemize}
  \item a deep-embedded first-order language for category theory, geared towards diagrammatic proofs, together with a generic formal definition of categorical diagrams, both based on a previous work by two of the authors~\cite{DBLP:conf/csl/MahboubiP24};
  \item automation support for proofs by duality in the corresponding reified proof system;
  \item automation support for the \emph{commerge} problem.
\end{itemize}
The corresponding code is available at the following url~\cite{repo}. The rest of the article is organized as follows. Section~\ref{sec:prelim} fixes some vocabulary and describes the corresponding formal definitions. Section~\ref{sec:dsl} describes the deep-embedded formalization of the first-order language, and of the related reified proof system. Section~\ref{sec:autocom} explains the algorithms involved in automating commutativity proofs. We conclude in Section~\ref{sec:concl} by discussing related work and a few perspectives.

\section{Preliminaries}\label{sec:prelim}

This section fixes some definitions and notations, and introduces their formalized counterpart when relevant. Some of them coincide with the preliminaries of our previous text~\cite{DBLP:conf/csl/MahboubiP24}, for the purpose of being self-contained. The description of their formalized counterpart is novel. We do not display implicit arguments in \Coq{} terms. Throughout this article, we use the word \emph{category} for 1-categories.

In all what follows, $\N \coloneqq \{0, 1, \dots \}$ refers to the set of non-negative integers, represented in \Coq{} by the type \coqescape{nat}, from its standard library. We also use the standard polymorphic type \coqescape{list} for finite sequences, equipped with the library on sequences distributed by the \MC{}~\cite{assia_mahboubi_2021_4457887} library. Some names thus slightly differ from those present in the standard library. For instance, the size $|l|$ of a finite sequence $l$, formalized as \coqescape{l : list T}, is \coqescape{size l}, instead of the standard \coqescape{length l}. We document in comments the definitions we use from this library when their names are not self-explanatory. We recall that \coqescape{b1 && b2} is the standard notation for the (boolean) conjunction of two boolean values \coqescape{b1 b2 : bool}. If $k\in\N$, then $[k]$ denotes the finite collection $\{0,\dots,k-1\}$, which is implemented by the sequence \coqescape{iota 0 n : list nat}.

\begin{definition}[General quiver, dual]\label{def:quiv}
  A \emph{general quiver} $\Q$ is a quadruple $(V_\Q, A_\Q, s_\Q\colon A_\Q \to V_\Q, t_\Q\colon A_\Q \to V_\Q)$ where $V_\Q$ and $A_\Q$ are two sets. The element of\/ $V_\Q$ are called the \emph{vertices} of $\Q$ and the element of $A_\Q$ are called \emph{arrows}. If $a \in A_\Q$, $s_\Q(a)$ is called the \emph{source} of $a$ and $t_\Q(a)$ is called its \emph{target}. The \emph{dual} of a quiver $\Q$ is the quiver $\Q^\dual \coloneqq (V_\Q,A_\Q,t_\Q,s_\Q)$, which swaps the source and the target maps of $\Q$.
\end{definition}

From now on, we casually call \emph{quivers} the special case of quivers with $V_\Q$ a finite subset of $\N$. The formal definition moreover assumes that vertices are labelled in order:
\begin{minted}{coq}
  (* Data of a quiver *)
  Record quiver : Type := quiver_Build {
    quiver_nb_vertex : nat;  (* number of vertices *)
    quiver_arc : list (nat * nat); (* sequence of arrows *) }.

  (* Well-formedness condition for quivers :
    all arrows involved in A have a source and target in bound *)
  Definition quiver_wf '(quiver_Build n A) : bool :=
    all (fun a => (a.1 < n) && (a.2 < n)) A.

  (* The dual quiver of a quiver, with same vertices and reversed arrows *)
  Definition quiver_dual '(quiver_Build n A) : quiver :=
    quiver_Build n (map (fun a => (a.2,a.1)) A).
\end{minted}
An arrow of a formalized quiver \coqescape{q : quiver} is thus given by an element in the sequence \coqescape{quiver_arc q}, itself a pair of integers giving its source and target respectively. Note that the index in the sequence matters, as the sequence may have duplicate. A formalized quiver is well-formed when the sources and targets of its arrows are in bound.

For the sake of readability, we use drawings to describe some quivers, as for instance:
\begin{figure}[h]
\begin{center}
\begin{tikzpicture}[inner sep=1pt, -latex,scale=.7]
  \node (0) at (0,0) {\textbf{.}};
  \node (1) at (1,0) {\textbf{.}};
  \node (2) at (2,0) {\textbf{.}};
  \scriptsize
  \draw (0) -- (1) node[midway, above] {};
  \draw (0) to[bend right=35] node[midway, above] {} (2);
  \draw (1) to[bend left=30] node[midway, above] {} (2);
  \draw (1) to[bend right=20] node[midway, above] {} (2);
\end{tikzpicture}
\end{center}
\end{figure}

For a quiver $\Q$ denoted by such a drawing, the convention is that $V_\Q = [\card{V_\Q}]$ and $A_\Q = [\card{A_\Q}]$, where $\card A$ denotes the cardinal of a finite set~$A$. From left to right, the drawn vertices correspond to $0, 1, \dots, \card{V_\Q}-1$. Arrows are then numbered by sorting pairs $(s_\Q, t_\Q)$ in increasing lexicographical order, as in:
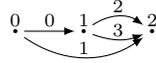
\begin{figure}[h]
\begin{center}
\begin{tikzpicture}[inner sep=1.5pt, -latex,scale=.9]
  \node (0) at (0,0) {\textbf{.}};
  \node (1) at (1,0) {\textbf{.}};
  \node (2) at (2,0) {\textbf{.}};
  \scriptsize
  \draw (0) node[above] {0};
  \draw (1) node[above] {1};
  \draw (2) node[above] {2};
  \draw (0) -- (1) node[midway, above] {0};
  \draw (0) to[bend right=35] node[midway, above=-.5pt] {1} (2);
  \draw (1) to[bend left=35] node[midway, above] {2} (2);
  \draw (1) to[bend right=20] node[midway, above=-.45pt] {3} (2);
  \end{tikzpicture}
\end{center}
\end{figure}

\begin{definition}[Morphism, embedding, restriction]
  A \emph{morphism of quivers} $m\colon \Q \to \Q'$, is the data of two maps $m_V\colon V_\Q \to V_{\Q'}$ and $m_A\colon A_\Q \to A_{\Q'}$ such that $m_V \circ s_\Q = s_{\Q'} \circ m_A$ and $m_V \circ t_\Q = t_{\Q'} \circ m_A$. Such a morphism is called an \emph{embedding of quivers} if moreover both $m_V$ and $m_A$ are injective. In this case we write $m\colon \Q \hookrightarrow \Q'$.
\end{definition}

We also use drawings to denote embeddings. The black part represents the domain of the morphism, the union of black and gray parts represents its codomain. Here is an example of an embedding of the quiver $\Qcobimap[]$ into the quiver drawn above.
\begin{center}
  \begin{tikzpicture}[inner sep=1pt, -latex,scale=.7]
    \node[black!30] (0) at (0,0) {\textbf{.}};
    \node (1) at (1,0) {\textbf{.}};
    \node (2) at (2,0) {\textbf{.}};
    \scriptsize
    \draw[black!30] (0);
    \draw (1);
    \draw (2);
    \draw[black!30] (0) -- (1);
    \draw[black!30] (0) to[bend right=35] (2);
    \draw (1) to[bend left=30] (2);
    \draw (1) to[bend right=20] (2);
  \end{tikzpicture}
\end{center}

For the purpose of this work, we actually only need to define formally embedding morphisms, called sub-quivers, which select the relevant vertices and arrows from a quiver:
\begin{minted}{coq}
  Record subquiver := subquiver_Build {
    subquiver_vertex : list nat; (* labels of the selected vertices *)
    subquiver_arc : list nat; (* indices of the selected arrows *) }.

  (* Performs the expected selection of vertices and arrows *)
  Definition quiver_restr '(subquiver_Build sV sA) : quiver -> quiver := (...)
\end{minted}

Here as well, restrictions of quiver only make sense under well-formedness conditions:
\begin{minted}{coq}
  (* Indices of the arrows to be selected are in bound *)
  Definition quiver_restr_A_wf sA '(quiver_Build n A) : bool :=
      all (gtn (size A)) sA.

  Definition quiver_restr_V_wf sV '(quiver_Build n A) : bool :=
    uniq sV && (* sV is duplicate-free *)
    all (gtn n) sV && (* all elements of sV are smaller than n *)
    all (fun a => (a.1 \in sV) && (a.2 \in sV)) A. (* any vertex involved in A is in sV*)

  (* Well-formed condition on the restriction of a quiver *)
  Definition quiver_restr_wf '(subquiver_Build sV sA) Q :=
    quiver_restr_A_wf sA Q && quiver_restr_V_wf sV (quiver_restr_A sA Q).
\end{minted}

\begin{definition}[Path-quiver]\label{def:pathquiv}
  The \emph{path-quiver of length $k$}, denoted by $\PQ_k$, is the quiver
  with $k+1$ vertices and $k$ arrows $([k+1], [k], \id, (i \mapsto i + 1))$.
\end{definition}

A path-quiver can be drawn as:
\begin{center} \tikz[baseline={([yshift=-.5ex]current bounding box.center)}, inner sep=1pt, -latex,scale=.7]{
  \bfseries \node (0) at (0,0) {.}; \node (1) at (1,0) {.}; \node (2)
  at (2,0) {.}; \node (3) at (3.6,0) {.}; \node at (2.8,0)
  {\normalfont $\dots$}; \draw (0) -- (1); \draw (1) -- (2); \draw
  (3.2,0) -- (3); \draw[-] (2) -- (2.4,0); }
\end{center}
with at least one vertex. Such a path-quiver is called \emph{nontrivial} if it has at least two vertices.

If $\Q$ is a general quiver, a morphism of the form $p\colon \PQ_k \to \Q$, for some $k$, is called a \emph{path of\/ $\Q$ from $u$ to $v$ of length $k$}, where $u \coloneqq p(0)$ and $v \coloneqq p(k)$. A general quiver is \emph{acyclic} if any path of this quiver is an embedding.

If $P$ is a nontrivial path-quiver, we define $\sot_P\colon \quiver[]{(0,0),(.2,0)}{} \hookrightarrow P$ to be the embedding mapping the first vertex on the leftmost vertex of $P$ and the second vertex on the rightmost vertex of $P$. Two paths $p_1\colon P_1 \to \Q$, $p_2\colon P_2 \to \Q$ of $\Q$ \emph{have the same extremities} if $p_1 \circ \sot_{P_1} = p_2 \circ \sot_{P_2}$. We denote by $\BP_{\!\Q}$ the set of pairs of paths of $\Q$ having the same extremities. Such a pair is called a \emph{bipath}.

Paths in a formal quivers are defined a ternary relation between two vertices and a sequence of arrows:
\begin{minted}{coq}
  (* Operations on lists:
    - (_ == _) is a generic boolean comparison test, in this case for lists of integers
    - rcons l x  is the list l followed by x
    - unzip[1 | 2] l is the list of fst (resp snd) elements of the list of pairs l
    - sub p A is the list of elements of A with index in p, in order *)
  Definition path (A : seq (nat * nat)) (u : nat) (p : seq nat) (v: nat) : bool :=
    (* all elements in p are in bound *)
    all (gtn (size A)) p &&
    (* p selects in A a list of adjacent arrows from u to v *)
    (u :: unzip2 (sub p A) == rcons (unzip1 (sub p A)) v).
\end{minted}

A path $p$ from a vertex $u$ to a vertex $v$ can thus be concatenated to a path $q$ from vertex $v$ to a vertex $w$: when endpoints are obvious, we just write $p \concat q$ the resulting path from $u$ to $w$. We sometimes abuse notations and write $e \concat p$ and $p \concat e$ when one of the paths contains a single arrow $e$.

We now introduce a special case of relations on pairs of paths with same extremities in a quiver, called \emph{path relations}. A path relation is an equivalence relation induced from a congruence on the corresponding free category to the quiver. Conversely, in a small category, the composition axiom induces a path relation on the underlying quiver. The formal definition of path relations is actually independent from that of quiver. A path relation is just a family of equivalence relations on sequences of integers, indexed by pairs of integers, that are compatible with the concatenation of paths:
\begin{minted}{coq}
  Record path_relation := {
    pi_r :> forall u v : nat, relation (seq nat) ;
    pi_equiv : forall u v, equivalence _ (pi_r u v) ;
    pi_cat_stable : forall u v w p p' q q',
        pi_r u v p p' -> pi_r v w q q' -> pi_r u w (p ++ q) (p' ++ q') }.
\end{minted}

Given a formal quiver and a path relation \coqescape{r : path_relation}, \coqescape{(pi_r r u v)} is expected to be a relation on the paths from vertex \coqescape{u} to vertex \coqescape{v}. The full relation on sequences, which relates any two sequences, can be used to complete a partial collection of relations, \eg{}, the path relation induced by a certain category on its underlying quiver, avoiding this way the need for otherwise cumbersome dependent types.

\section{A two-level approach}\label{sec:dsl}

\subsection{Formulas and diagrams}\label{ssec:formdiag}

Paraphrasing Mc Lane~\cite{zbMATH01216133}, many properties of category theory can be \enquote{unified and simplified by a presentation with diagrams of arrows}. Categorical diagrammatic reasoning consists in transforming a proof of category theory into a proof about some quivers, decorated with the data of a certain category. Actually, once the appropriate quivers are drawn, the data themselves can be forgotten, but for the induced path relation, which is the only relevant information for a diagrammatic proof. In~\cite{DBLP:conf/csl/MahboubiP24}, we proposed a multi-sorted first-order language for category theory, geared towards diagrammatic reasoning: following the structure of a formula in this language constructs the quivers associated with the corresponding statement of category theory, encoded in the sorts of the variables. We recall its definition:
\begin{definition}\label{def:sig}
    We define a many-sorted signature $\Sigma$ with sorts the collection of finite acyclic quivers. Signature\/ $\Sigma$ has one function symbol $\restr_{m\colon \Q'\hookrightarrow \Q}$, of arity $\Q \to \Q'$, per each \emph{injective} quiver morphism $m\colon \Q' \hookrightarrow \Q$ between two quivers $\Q$ and $\Q'$, and one predicate symbol $\com_\Q$, on sort $\Q$, for each finite acyclic quiver $\Q$.
  \end{definition}

\begin{example}
Writing the sorts of quantified variables as a subscript of the quantifier, here is for instance a predicate of arity $\Qmap[] \times \Qmap[] \times \Qmap[]$ describing composite of arrows:
\begin{align*}
  \Comp(x,y,z)\colon\qquad& \exists_{\Qcomp}\, w,\quad \restr_{\mapicomp}(w) \eqd x \ \wedge \ \restr_{\mapiicomp}(w) \eqd y \\
  &\qquad \wedge \ \restr_{\mapiocomp}(w) \eqd z \ \wedge \ \com(w)
\end{align*}
\end{example}

\begin{example}\label{ex:mono}
  Here is a predicate of arity $\Qmap[]$ describing monomorphisms:
  \begin{align*}
    \Mono(x)\colon\qquad& \forall_{\Qmono}\, w,\quad \restr_{\mapioomono}(w) \eqd x \ \Rightarrow \ \com(\restr_{\compimono}(w))  \\
    &\qquad \Rightarrow \com(\ \restr_{\compiomono}(w)) \ \Rightarrow \ \com(\restr_{\cobimapimono}(w))
  \end{align*}
\end{example}

\begin{listing}[!t]
\begin{minted}[mathescape=false]{coq}
  Inductive term :=
      | Var of nat (* variable, named with an integer, $k denotes term (Var k) *)
      | Restr of subquiver & term. (* the 'restr' symbol *)

  Inductive formula :=
      | Forall of quiver & formula
      | Exists of quiver & formula
      | Imply of formula & formula (* Denoted with infix symbol -=> *)
      | And of formula & formula
      | FTrue (* Top atom *)
      | Commute of term (* the 'commute' predicate symbol *)
      | EqD of term & term.  (* the equality predicate symbol*)
\end{minted}
\caption{Terms and formulas}
\label{lst:formulas}
\end{listing}

Listing~\ref{lst:formulas} is the formalized counterpart of Definition~\ref{def:sig}. A term \coqescape{t : term} is thus either a variable \coqescape{Var n}, named with a natural number \coqescape{n}, or of the form \coqescape{Restr m t} for \coqescape{t} a term and \coqescape{m} a sub-quiver, seen as a morphism of quivers. Observe that the source quiver of this morphism is left undefined -- it only becomes explicit when the term is evaluated. The type \coqescape{formula} defines a first-order logic whose atoms stand for equality, commutativity or true. Quantifiers bind de Bruijn indexes and are annotated with a quiver, the sort of the bound variable. To be consistent with the theory $\Sigma$, one should only use acyclic quivers. Computing the sort of a term thus requires first annotating each of its variables with a sort. The sort of  a term of the form \coqescape{Restr m t} is then the quiver obtained by restricting the sort of term \coqescape{t} using \coqescape{m}.

Given a list \coqescape{l} of quivers, providing a sort to each of its variable, a term \coqescape{t : term} is well-formed in this context, written \coqescape{term_wf l t}, if \coqescape{l} is long enough and all the subterms of \coqescape{t} have a well-formed sort. In a closed formula, the sort of a variable is read on the corresponding quantifier. More generally, a formula \coqescape{f : formula} is well-formed in a context \coqescape{l}, written \coqescape{formula_wf l f}, if \coqescape{l} is long enough to provide a sort to each free variable in \coqescape{f} and if all the terms appearing in \coqescape{f} have a well formed sort.

Here is the corresponding formal predicate to Example~\ref{ex:mono}.
\begin{minted}[mathescape=true]{coq}
  (* We use a notation for easing the definition of quivers
  without isolated vertices from their arcs*)
  Definition monoQ : quiver := {Q [:: (0,1);(0,1);(0,2);(1,2)]}. (* This is $\raisebox{4pt}{\Qmono}$*)
  Definition mapQ : quiver := {Q [:: (0,1)]}. (* This is $\Qmap[]$ *)

  (* Lambda_arc constructs a predicate from a sequence of quivers and a formula, the first
  argument is the arity, providing sorts for the free variables of the second, in order.*)
  Definition monoF : predicate :=
    Lambda_arc [:: mapQ] (* the one-element arity sequence *)
    (Forall monoQ (
      EqD (Restr {sA [:: 3]} $0) $1
      -=> Commute (Restr {sA [:: 0 ; 2 ; 3]} $0)
      -=> Commute (Restr {sA [:: 1 ; 2 ; 3]} $0)
      -=> Commute (Restr {sA [:: 0 ; 1]} $0))) .
\end{minted}

The interpretation of a term \coqescape{f : formula} as a \Coq{} statement, in sort \coqescape{Prop}, is relative to a formal notion of \emph{diagram}, which is by definition an instance of the following structure \coqescape{diagram_type}:
\begin{minted}{coq}
  Record diagram_package (diagram : Type) := diagram_Pack {
    diagram_to_quiver : diagram -> quiver; (* underlying quiver of a diagram *)
    diagram_restr : subquiver -> diagram -> diagram; (* restriction *)
    eqD : equivalence diagram; (* setoid relation on type diagram *)
    eq_comp : diagram -> path_relation; (* path relation *) }.

  Structure diagram_type := diagram_type_Build {
    diagram_sort :> Type; (* a diagram_type coerces to its carrier type *)
    diagram_to_package :> diagram_package diagram_sort; }.
\end{minted}

A diagram is thus a term in the carrier type of an instance of structure \coqescape{diagram_type}, which can be seen as a model. A diagram commutes when the associated path relation is \emph{full}, \ie, any two paths in the underlying quiver with same source and target are related:
\begin{minted}{coq}
  Definition commute (d : diagram_type) (D : diagram) :=
    path_total (diagram_sort D) (eq_comp D).
\end{minted}

Formal diagrams can be defined from formalized categories, but not only. For instance, one can define an instance of \coqescape{diagram_type} with the following carrier type:
\begin{minted}{coq}
  Record zmod_diagram : Type := ZModDiagram {
    zmob : nat -> zmodType;
    zmmap : forall u v : nat, nat -> {additive zmob u -> zmob v};
    zmdiag_to_quiver : quiver }.
\end{minted}
where \coqescape{zmodType} is a structure for Abelian groups in the \MC{} library, and \coqescape{{{additive A -> B}}} is the type of morphisms between two Abelian groups \coqescape{A} and \coqescape{B}. We can now explain how to turn a term \coqescape{f : formula} into a \Coq{} statement, given a sequence of diagrams:
\begin{minted}{coq}
  Fixpoint formula_eval (d : diagram_type) (stack : list d) (f : formula) : Prop :=
    match f with
    | Forall Q f => forallD D :: diagram_on Q, formula_eval d (D :: stack) f
    | Exists Q f => existsD D :: diagram_on Q, formula_eval d (D :: stack) f
    | Imply f1 f2 => formula_eval d stack f1 -> formula_eval d stack f2
    | And f1 f2 => formula_eval d stack f1 /\ formula_eval d stack f2
    | FTrue => True
    | Commute t => if term_oeval stack t is Some D then commute D else False
    | EqD t1 t2 =>
      match term_oeval stack t1, term_oeval stack t2 with
      | Some DG1, Some DG2 => eqD DG1 DG2
      | _, _ => False
      end
    end.
\end{minted}
where the notations \coqescape{forallD D :: diagram_on Q, P} and \coqescape{existsD D :: diagram_on Q, P} bind variable \coqescape{D} in \coqescape{P}, so as to quantify \coqescape{P} over diagrams \coqescape{D : d} with underlying quiver \coqescape{diagram_to_quiver D} equal to \coqescape{Q}. The evaluation \coqescape{term_oeval stack t : option d} of a term \coqescape{t} in context \coqescape{stack} defaults to \coqescape{None} when the context \coqescape{stack} is too small and otherwise computes, when possible, the prescribed restriction of the diagrams before the evaluation of the atom using the relevant commutativity and equality predicates.

\subsection{Structural duality}\label{ssec:dual}

As briefly alluded to in the conclusion of our previous article~\cite{DBLP:conf/csl/MahboubiP24}, it is possible to prove a duality theorem at the meta-level of the deep-embedded first-order language. The variant presented below is updated to the current, bundled, representation of diagram types, and to a slightly different, albeit equivalent, definition of models. For any formula \coqescape{f : formula}, we define its dual formula \coqescape{formula_dual f} by structural induction, dualizing all the quivers involved in \coqescape{f}. For the sake of readability, the code uses a few notations and coercions:
\begin{minted}{coq}
    (* We fix a type for diagrams until the end of the section. *)
    Variable d : Type.

    (* A few local notations to ease reading *)
    Local Notation model := diagram_package.
    Local Notation model_dual := dual_diagram_Pack.

    (* Coercion "conflating" a model (a diagram_type) with carrier type d and its
    diagram_package. *)
    Local Coercion diagram_type_of := (@diagram_type_Build d).
\end{minted}

Any model of diagrams \coqescape{M : model d} has a dual \coqescape{model_dual M : model d}, obtained from \coqescape{M} by keeping the same data, but dualizing its quiver and reversing its path relation. We can now prove the property \coqescape{formula_eval_duality}, stated in Listing~\ref{lst:duality}, characterizing the evaluation of this dual formula in the dual of a diagram. Note that a coercion is hidden in the type of the argument \coqescape{ctx} of \coqescape{formula_eval}, to the common carrier type of \coqescape{M} and \coqescape{model_dual M}.

As a corollary, if a formula holds for all diagrams in a certain diagram type, then so does its dual. The variant \coqescape{duality_theorem_with_theory} is equally direct but slightly more interesting, as dependent type \coqescape{P} shall be used to describe a specific class of models, \eg{} models of a given theory, provided that their description is \enquote{auto-dual}, cf. Listing \ref{lst:duality}.
\begin{listing}[h!]
\begin{minted}{coq}
Theorem formula_eval_duality (M : model d) (ctx : seq d) (f : formula) :
  @formula_eval (model_dual M) ctx (formula_dual f) <-> @formula_eval M ctx f.

(* If a formula is valid in every model, then so is the dual formula *)
Corollary duality_theorem (ctx : seq d) (f : formula) :
  (forall M : model d, @formula_eval M ctx f) ->
  forall M : model d, @formula_eval M ctx (formula_dual f).

(* Relativized variant to a specific class P of models *)
Corollary duality_theorem_with_theory (ctx : seq d) (f : formula) (P : model d -> Prop) :
  (forall M, P M -> P (model_dual M)) -> (forall M, P M -> @formula_eval M ctx f)
  -> forall M, P M -> @formula_eval M ctx (formula_dual f).
\end{minted}
\caption{Duality theorems}\label{lst:duality}
\end{listing}

\subsection{Proofs}

The deep-embedded level also features a data-structure \coqescape{valid_proof} for (deep-embedded) proofs of deep-embedded formulas, and implements a checker \coqescape{check_proof} for these proofs. A correctness theorem ensures that for any well-formed deep-embedded formula \coqescape{f : formula}, a positive answer of the proof checker entails the provability of the interpretation of \coqescape{f} in any model \coqescape{d}:
\begin{minted}{coq}
  Theorem check_proof_valid (d : diagram_type) (f : formula) (pf : valid_proof) :
    formula_wf [::] f -> check_proof f pf = true -> formula_eval d [::] f.
\end{minted}

The deep-embedded level also features a data-structure \coqescape{sequent}, used to reify a proof in progress, and a type \coqescape{tactic} for actions making progress in a proof:
\begin{minted}[autogobble]{coq}
  Inductive sequent := sequent_Build {
    context : seq quiver;
    premises : seq formula;
    goal : formula; }.

  Definition tactic := sequent -> option sequent.
\end{minted}
Note that there is only one goal attached to a sequent, as disjunction is not part of our formal language. To a sequent with context \coq{[:: Q_1 ; ... ; Q_n ]}, premises \coq{[:: H_1 ; ... ; H_m ]} and goal \coq{G}, one can associate the following term in type \coqescape{formula}:
\begin{minted}[autogobble]{coq}
  Forall Q_1 ( ... ( Forall Q_n (
    Imply   (And H_1 ( ... (And H_m Ftrue) ... ))   G )) ... )
\end{minted}
The sequent is well-formed if the corresponding formula is a well-formed formula. In the other direction, to any \coqescape{f : formula}, one can associate the sequent with goal \coqescape{f} and with empty context and no premise. A tactic $\tau$ is \emph{valid} when for any well-formed sequent $s$, if $\tau\,s$ is some $s'$, then $s'$ is also well-formed and the evaluation of $s'$ implies the evaluation of $s$. In this context, a \emph{valid proof} is just a list of valid tactics.

To check that a valid proof effectively provides a proof of some given formula $f$, one can perform the following steps. First, compute the sequent associated to the formula. Then, for each tactic in the proof, apply the tactic. If at some point the tactic returns \coq{None} the proof is not correct. Otherwise, check that the goal of the final sequent is \coq{FTrue}. If this is the case, then the evaluation of $f$ is true. In the implementation, the verification is performed by the function \coq{check_proof}, and the conclusion is proven in Theorem \coq{check_proof_valid}.

The next step is to implement a set of \emph{useful} tactics, and to prove that they are valid. This has been done for basic tactics like introduction and elimination rules, or for more involved tactics like the \mintinline{c}`Rewrite` tactic, or the \coq{Comauto} tactic, for automating the proof of commutative atoms. Starting from the rules of the proof system, more complex tactics combine existing ones in a relevant way, so as to considerably reduce the size of the proofs. In order to be valid, some tactics may also require more assumptions from the model (the type of diagrams) used to evaluate the formulas.

Here is for instance the statement of the formula corresponding to Lemma~\ref{lem:monomonom}, using an infix notation \coqescape{-=>} for the \coqescape{Imply} constructor of type \coqescape{formula}:
\begin{minted}{coq}
  (* when the quiver of a formula have no isolated vertex, formula_fill_vertices
  allows for a shorter description of quivers, only by their arcs. *)
  Definition mono_monomPF : formula :=
  formula_fill_vertices [::] (
  Forall compQ (
    Commute $0 -=> monoF App (Restr {sA [:: 1]} $0) -=> monoF App (Restr {sA [:: 0]} $0))).
\end{minted}

The reified proof of this statement is currently done by applying successively twelve tactics.\footnote{See \texttt{mono\_monom\_pf} in file \texttt{diagram\_chasing/mono\_monom.v}}

In fact, duality arguments are implemented by instrumenting proofs so as to check that they are amenable to duality arguments, which is more convenient in practice than the structural argument described in Section~\ref{ssec:dual}. We thus gather a tactic $\tau$ and its dual tactic $\tau^*$ such that for any sequent $s$, the dual of $\tau s$ is equal to $\tau^*$ applied to the dual of~$s$. Very often, a tactic and its dual are just identical. A second theorem \coqescape{duality_theorem}, not to be confused with that of Section~\ref{ssec:dual}, ensures that, indeed, if every tactic of a proof of some formula has such a dual tactic, then the proof obtained by taking the dual tactics will be a proof of the dual formula.\footnote{See file \texttt{diagram\_chasing/FanL.v}}

\begin{minted}{coq}
  (* Biproofs are pairs of proofs of same size *)
  Structure biproof := biproof_Build {
    biproof_primal : proof;
    biproof_dual : proof;
    biproof_eq_size : size biproof_primal == size biproof_dual; }.

  Theorem duality_theorem f (bpf : biproof) :
  (* assuming that the tactics in the pair of proofs bpf are pairwise dual *)
    biproofD bpf ->
  (* then the primal proof proves a formula iff the dual proof proves its dual *)
    check_proof (formula_dual f) (biproof_dual bpf) =
    check_proof f (biproof_primal bpf).
\end{minted}

All the tactics currently implemented have a dual tactic (though for some of them, the proof of their validity has not yet been written). A duality argument can be used to prove the dual statement to Lemma~\ref{lem:monomonom}, and we provide the corresponding deep-embedded formula.\footnote{See file \texttt{tests\_and\_examples/mono\_monom.v}}

\begin{lemma}\label{lem:epimepi}
  Let $\Cat$ be a category. For any morphism $f$ and $g$ such that $g \circ f \in Hom(\Cat)$, if $g \circ f$ is an epimorphism, then so is $g$.
\end{lemma}

\begin{proof}
From Lemma~\ref{lem:monomonom}, by duality.
\end{proof}

\begin{minted}{coq}
  Definition epiF :=
    Lambda_arc [:: mapQD]
    (Forall monoQD (
      EqD (Restr {sA [:: 3]} $0) $1
      -=> Commute (Restr {sA [:: 0 ; 2 ; 3]} $0)
      -=> Commute (Restr {sA [:: 1 ; 2 ; 3]} $0)
      -=> Commute (Restr {sA [:: 0 ; 1]} $0))).

  Definition epi_mepiPF : formula :=
  formula_fill_vertices [::] (
  Forall compQD (
    Commute $0 -=> epiF App (Restr {sA [:: 1]} $0) -=> epiF App (Restr {sA [:: 0]} $0))).
\end{minted}
which is similar to term \coqescape{monoPF} and formula \coqescape{mono_monomPF} except that quivers \coqescape{mapQ}, \coqescape{monoQ} and  \coqescape{compQ} have been dualized, respectively into \coqescape{mapQD}, \coqescape{monoQD} and \coqescape{compQD}. The reified proof takes a single tactic, which dualizes the proof of the statement on monomorphisms.

\section{Automating commutativity proofs}\label{sec:autocom}

In order to apply most lemmas obtained by diagram chasing, one first has to prove that a certain diagram is commutative. For instance, the so-called \emph{five lemma}, which allows to prove that some map is an isomorphism, requires to have a commutative diagram over the quiver of Figure \ref{fig:five-lemma}. Yet proving that such a diagram is commutative by checking one equality per bipaths in this diagram would be excessively tedious. Commutativity of a larger diagram is typically obtained from the commutativity of certain sub-diagrams, and the proof of this implication is often little detailed, or not at all. For instance, in the case of Figure \ref{fig:five-lemma}, it actually suffices to check four equalities, one for each sub-square of the quiver. More precisely, each pair of paths going from the top-left corner to the bottom-right corner must correspond to equal morphisms. Once these four equalities has been proven, we infer that each square commutes.\footnote{See file \texttt{diagram\_chasing/Bipath.v} for a formalized proof.}
\begin{figure}[h]
  \begin{center}
    \begin{tikzcd}[column sep=scriptsize]
      A \rar\dar & B \rar\dar & C \rar\dar & D \rar\dar & E \dar \\
      A'    \rar & B'    \rar & C'    \rar & D'    \rar & E'
    \end{tikzcd}
  \end{center}
  \caption{The five-lemma diagram. \label{fig:five-lemma}}
\end{figure}

In section~\ref{ssec:commerge}, we describe an algorithm for deciding the commerge problem for diagrams with acyclic underlying quivers, so as for instance to automate the proof that the diagram of Figure \ref{fig:five-lemma} commutes as soon as the aforementioned four bipaths commute. In section~\ref{ssec:comcut}, we describe a heuristic for discovering a collection of sub-diagrams whose commutativity entails that of a larger one.

\subsection{Decision procedure for the commerge problem}\label{ssec:commerge}

In~\cite{DBLP:conf/csl/MahboubiP24}, we provided a pen-and-paper proof of the decidability of the commerge problem for diagrams with acyclic underlying quiver, and the undecidability of its generalization to possibly cyclic underlying quivers. In this section, we describe the more practical algorithm implemented by the tactic \coqescape{Comauto}. In particular, this implementation comes with a formal proof of correctness, which is the main ingredient in the validity proof of the corresponding tactic.

The algorithm operates on an acyclic quiver $\Q$ and a finite collection $\Q_1, \dots, \Q_k$ of subquivers of $\Q$, representing the commutativity assumptions. Denote $l$ the union $\bigcup_{i}\BP_{\Q_i}$, slightly abusing notations by denoting $\Q_i$ the quiver induced by the corresponding subquiver. The algortithm checks whether the smallest path relation $\cl_\Q(l)$ induced by $l$ is full, that is $\cl_\Q(l) = \BP_\Q$. Without loss of generality, we can assume that the acyclic quiver $\Q$ is topologically sorted in reverse order, that is, that any path in $\Q$ follows a sequence of vertices with decreasing labels. The implementation actually performs a topological sort of the quiver~\footnote{See \texttt{diagram\_chasing/TopologicalSort.v}} and updates the representation of $\cl_\Q(l)$ accordingly.

For any given vertices $u$ and $v$ in $V_\Q$, we introduce $\arcsfrom{u,v} \subset A_\Q$ the set of adjacent arrows of $u$ starting a path to $v$ in $\Q$ and $\arcsto{v,u} \subset A_\Q$ that of adjacent arrows of $v$ ending a path from $u$ in $\Q$:
\begin{align*}
\arcsfrom{u,v} & \triangleq \{ e \in A_\Q\ |\ \exists p,\ e\concat p \textrm{ is a path from } u \textrm{ to } v \} \\
\arcsto{v,u} & \triangleq \{ e \in A_\Q\ |\ \exists p,\ p\concat e \textrm{ is a path to } v \textrm{ from } u \}
\end{align*}
We define $\G_{u,v}$ the multigraph whose vertices are the elements of $\E_{u,v} \triangleq \arcsfrom{u,v} \cup \arcsto{v,u}$. An arc of $\G_{u,v}$ relates vertices $e_1, e_2 \in \E_{u,v}$ when:
\begin{itemize}
  \item either there is a path in $\Q$ from $u$ to $v$ containing both $e_1$ and $e_2$;
  \item there exists $i$ and $p_1, p_2$ paths in $\Q_i$ from $u$ to $v$ such that $e_1$ (resp. $e_2$) appears in $p_1$ (resp. $p_2$).
\end{itemize}
We moreover introduce the binary relation $R_{u,v}$, on the elements of $\E$: for any $e_1, e_2 \in \E$, $R_{u,v}(e_1, e_2)$ holds if and only if any path $p_1$ from $u$ to $v$ containing $e_1$ is related by $\cl_\Q(l)$ to any  any path $p_2$ from $u$ to $v$ containing $e_2$. If $R_{u,v}$ is full, then $\cl_\Q(l)$ contains all the pairs of paths from $u$ to $v$.
\begin{lemma}\label{lem:commerge}
  For any vertices $u,v \in V_\Q$, if\/ $\G_{u,v}$ is connected, then $R_{u,v}$ is full.
\end{lemma}
The decision procedure constructs the multigraphs $\G_{u,v}$ for any pair of vertices, and checks that they are all connected. If so, then $\cl_\Q(l)$ contains $\BP_\Q$.

\begin{proof}
  We prove Lemma~\ref{lem:commerge} for vertices respectively labelled $u + n$ and $u$, for a vertex $u \in V_\Q$, by induction on $n$. When $n = 0$, then the result holds because paths are empty.

  We now assume that the result holds for any $u$ and any $k < n$, and that $\G_{u+n,u}$ is connected. Observe first that as a consequence of the induction hypothesis, and because the graph is  topologically sorted in reverse order, any two paths from $u + n$ to $u$ sharing their initial or their final arrow are in $\cl_\Q(l)$. We now fix $u\in V_\Q$ and $e_1, e_2 \in \E_{u+n,u}$ and we need to prove that $R_{u+n,u}(e_1, e_2)$. We proceed by induction on the length of a path in $\G_{u+n,u}$ relating $e_1$ and $e_2$. If this path is empty, then $e_1 = e_2$ is either an element of $\arcsfrom{u,v}$ or of $\arcsto{v,u}$ and we can conclude using the initial observation. We now suppose that there is an arc $e\in \E_{u+n,u}$ and a path $t$ in $\G_{u+n,u}$ such that $(e_1, e)$ is an edge of $\G_{u+n,u}$ and $t$ is a path from $e$ to $e_2$ in $\G_{u+n,u}$. Let $p_1$ (resp. $p_2$) be a path in $\Q$ containing $e_1$  (resp. $e_2$). If there is a path $p$ in $\Q$ from $u$ to $v$ containing both $e$ and $e_1$ then $p$ is related to $p_1$, by the observation, as the two paths share either their initial or their final arrow. But $p$ is also related to $p_2$ by the (second) induction hypothesis as $p$ contains $e$ and $p_2$ contains $e_2$. The conclusion follows by transitivity of the path relation. Now suppose that $e$ and $e_1$ respectively belong to $q$ and $q_1$, paths from $u+n$ to $u$ in some $\Q_i$. Paths $q_1$ and $q$ are related, by definition of $\cl_\Q(l)$. But $q$ and $p_2$ are also related by the (second) induction hypothesis, as they respectively contain $e$ and $e_2$. The conclusion follows by transitivity.
\end{proof}

\subsection{Finding sufficient commutativity conditions}\label{ssec:comcut}

In fact, one can even use the computer to guess a sufficient list of equalities entailing that a certain diagram commutes, instead of providing it explicitly by hand. In this section, we explain how to do so in practice. The corresponding algorithm has been implemented under the name \emph{comcut}. Although not strictly needed for the purpose of a tactic, its correctness has been proven formally.\footnote{See the folder \texttt{comcut}}.

Let $\Q$ be a quiver with vertices $V=[n]$, for some $n\in\N$, and arcs $A_\Q$. Moreover, we assume that $\Q$ is topologically sorted in a reversed order, \ie, if an arc goes from $u$ to $v$, then $u > v$. From such an input, \emph{comcut} returns a list of bipaths $l \subseteq \BP_\Q$ such that $\cl_\Q(l)$, that is the smallest path relation in $\Q$ containing $l$, is equal to $\BP_\Q$.

The algorithm works by induction on the size of $\Q$. Let $u_0 := n-1$ be the top vertex. If there is no arc whose source is $u_0$, then we just have to apply \emph{comcut} to $\Q$ deprived from the vertex $u_0$ (note that $u_0$ cannot be a target). Otherwise, let $a_0 \in A$ be an arc with source $u_0$. Denote by $v_0$ the target of $a_0$. Consider the quiver $\Q'$ obtain from $\Q$ by removing the arc $a_0$. Applying the algorithm to $\Q'$, we get a list $l' \subset \BP_{\Q'}$ such that $\cl_{\Q'}(l')=\BP_{\Q'}$. We complete this list to a list $l$ as follows. Let $\~W := \acc(u_0)\cap\acc(v_0)$ be the set of vertices accessible both from $u_0$ and from $v_0$ in $\Q'$. Set
\[ W := \{ w \in \~W \mid \forall w'\in\~W\setminus\{w\}, w \notin \acc(w') \}. \]
For each $w \in W$, select a path $p_w$ from $u_0$ to $w$ in $\Q'$ and a path $q_w$ from $v_0$ to $w$ in $\Q'$. Set $l := l' \cup \{ (p_w, a_0 \concat q_w) \mid w \in W \}$.

\begin{proposition}
  The list $l$ constructed above verifies $\cl_\Q(l) = \BP_\Q$.
\end{proposition}

Before proving the proposition, let us quickly explain how to get an effective implementation from the above description. To be able to compute accessibility and to reconstruct the different paths, one computes a square matrix indexed by vertices whose $(u,v)$ entry is either empty if there is no nontrivial path from $u$ to $v$, or contains an arc $a$ such that there is a path from $u$ to $v$ which starts by $a$. To update such a matrix for the quiver $\Q'$ into a matrix corresponding to the quiver $\Q$, it suffices to set $a_0$ to all the entries of the form $(u_0,v)$ for $v \in \acc(v_0)$. The rest of the algorithm is easy to write down.

\begin{proof}
  Let $p, q$ be paths from $u$ to $v$ in $\Q$. Note that the arc $a$ can only appear as the first element of $p$ and $q$. If $a$ does not appear in $p$ nor in $q$, or if it appears in both, then $(p,q)$ was already in $\BP_\Q(l')$. Hence, up to symmetry, it remains the case where $u=u_0$, $p$ belongs to $\Q'$ and $q=a_0 \concat q'$ with $q'$ a path of $\Q'$. By definition of $\~W$, $v \in \~W$. Let $w \in W$ such that $v$ is accessible from $w$. Let $r$ be a path from $w$ to $v$ (cf. Figure \ref{fig:two_bifurcations}).
  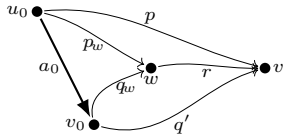
\begin{figure}[h]
    \begin{center}
    \begin{tikzpicture}[rotate=90, xscale=-1, scale=1.5]
      \scriptsize

      \begin{scope}
        \draw (0,0)   node[inner sep = .5pt] (a1) {\rawvertex{}} node[left] {$u_0$};
        \draw (1,-.5) node[inner sep = .5pt] (a2) {\rawvertex{}} node[left] {$v_0$};
        \draw (.5,-1) node[inner sep = .5pt] (w)  {\rawvertex{}} node[below]   {$w$};
        \draw (.5,-2) node[inner sep = .5pt] (v)  {\rawvertex{}} node[right]  {$v$};
      \end{scope}

      \begin{scope}[ultra thin, ->, inner sep=1.5pt]
        \draw[thick, -Latex] (a1) -- (a2) node[midway, left] {$a_0$};
        \draw (a1) to[out=-100, in=100] node[midway, above]    {$p$}   (v);
        \draw (a1) to[out=-85, in=110]  node[midway, below]   {$p_{\!w}$} (w);
        \draw (a2) to[out=-80, in=70]   node[midway, below]   {$q'$}  (v);
        \draw (a2) to[out=-180, in=60]  node[near end, below]   {$q_{\!w}$} (w);
        \draw (w) to[out=-110, in=80]   node[midway, below]   {$r$}   (v);
      \end{scope}
    \end{tikzpicture}
  \end{center}
    \caption{Decomposition of the relation between two paths thanks to an element $w \in W$. \label{fig:two_bifurcations}}
  \end{figure}
  Then, $(p, p_w \concat r)$ and $(q_w \concat r, q')$ both belong to $\BP_{\Q'}$, and $(p_w, a_0 \concat q_w)$ belongs to $l$. Hence we get the following sequence of relations in $\cl_\Q(l')$.
  \[ p \ \sim \ p_w \concat r \ \sim \ a_0 \concat q_w \concat r \ \sim \ a_0 \concat q' = q, \]
  which proves the proposition.
\end{proof}

\begin{remark}
Note that it may happen that $l$ is not minimal, see Figure \ref{fig:counter-example_comcut} for a counterexample.
\begin{figure}[h]
  \begin{center}
    \begin{tikzpicture}[rotate=90, xscale=-1, scale=1.5]
      \scriptsize

      \node (3) at (0,0)    {$3$};
      \node (4) at (-30:-1) {$4$};
      \node (0) at (30:1)   {$0$};
      \node (1) at (80:1.5) {$1$};
      \node (2) at (-80:.8) {$2$};

      \foreach \i/\j in {4/3,4/2,4/1,4/0,3/2,3/1,3/0,2/0,1/0} {
        \draw[->] (\i) -- (\j);
      }
    \end{tikzpicture}
  \end{center}
  \caption{Counterexample to the minimality of the comcut algorithm. Before adding the arc $(4,3)$, we need four equalities to ensure the commutativity. Adding $(4,3)$ forces to add two more relations, but one of the previous relations becomes useless. \label{fig:counter-example_comcut}}
\end{figure}
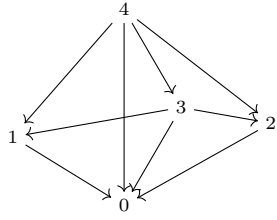
\end{remark}

\section{Conclusion}\label{sec:concl}

We have described a first step towards the implementation of a generic library for writing reliable categorical diagrammatic proofs, available online~\cite{repo}. Such a library can serve two purposes. The first one is to assist mathematician authors in writing reliable proofs, the other is to provide the mandatory infrastructure for expanding the existing corpus of formalized category theory, but also of formalized algebraic topology, and homological algebra. As expressed by the \textsc{Mathlib} community~\cite{mathlibzulip}, the lack for such a tool is major showstopper for the latter. However, the current state of the present library arguably only provides a low level language for categorical statements and the next steps should enrich the collection of formula combinators, \eg{} for limits, pullbacks, etc. as well as the gallery of diagram models.

Independence from any library of category theory is achieved by hosting a dedicated proof system inside that of a proof assistant, \Coq in this case, following the classic formalization technique of \emph{deep-embedding}~\cite{DBLP:conf/tpcd/BoultonGGHHT92}. Dependent types allow to formalize a structural duality property for this language. We are not aware of any comparable formal-proof-producing automation tactics for proving the commutativity of diagrams, nor for performing duality arguments. However, some existing libraries of formalized category theory, notably \textsf{Mathlib}, for the \textsf{Lean} proof assistant, and \textsf{Unimath}, a \Coq{} library for univalent mathematics, provide some tools to ease proofs by diagram chasing, either with brute-force rewrite-based tactics, or with a graphical editor for generating proof scripts~\cite{lafont}. Other existing libraries of formalized category theory, including Jacobs and Timany's~\cite{DBLP:conf/rta/Timany016}, do not include any specific support for diagrammatic proofs. We refer the interested reader to the later article for a survey of existing libraries of formalized category theory, which remains quite relevant for the purpose of this discussion, to the notable exception of the more recent \textsf{Mathlib} chapter on category theory. The later serves as a basis for Himmel's formalization of abelian categories in \textsf{Lean}~\cite{himmel}, including proofs of the five lemma and of the snake lemma, and proof (semi-)automation tied to this specific formalization. Duality arguments are not addressed. Also in the \textsf{Mathlib} ecosystem, Monbru~\cite{monbru} also discusses automation issues in diagram chases, and provides heuristics for generating them automatically, albeit expressed in a pseudo-language.

Other computer-aided tools exist for diagrammatic categorical reasoning, with a specific emphasis on the graphical interface. Notably, the accomplished \textsf{Globular} proof assistant~\cite{DBLP:journals/lmcs/BarKV18} stems from similar concerns about the reliability of diagrammatic reasoning, but for higher category theory. It is  geared towards visualization rather than formal verification and implements various algorithms for constructing and comparing diagrams in higher categories. Barras and Chabassier have designed a graphical interface for diagrammatic proofs which also provides a graphical interface for generating \Coq{} proof scripts of string diagrams, and visualizing \Coq{} goals as diagrams~\cite{barras-chabassier}. But up to our knowledge, this tool does not include any specific automation.

\bibliography{biblio}
\end{document}